\documentclass[12pt]{article}
\usepackage{amsmath}
\usepackage{amsthm}	
\usepackage{amsfonts}	
\usepackage{graphicx}	
\usepackage{color}

{\theoremstyle{plain}%
  \newtheorem{thm}{\sc Theorem}[section]%
  \newtheorem{prop}[thm]{\sc Proposition}%
  \newtheorem{cor}[thm]{\sc Corollary}%
}
{\theoremstyle{remark}
  \newtheorem{rem}[thm]{\sc Remark}
}

\begin{document}

\title{Resonant-tunneling\\ in discrete-time quantum walk}%
\markright{Resonant-tunneling in DTQW}

\author{%
  {Kaname Matsue}\thanks{Institute of Mathematics for
    Industry/International Institute for Carbon-Neutral Energy
    Research (WPI-I$^2$CNER), Kyushu University, Fukuoka 819-0395,
    Japan,
    kmatsue@imi.kyushu-u.ac.jp
  }%
  \and %
  {Leo Matsuoka}\thanks{
    Graduate School of Engineering, Hiroshima University,
    Higashi-hiroshima 739-8527, Japan,
    leo-matsuoka@hiroshima-u.ac.jp
  }%
  \and %
  {Osamu Ogurisu}\thanks{Division of Mathematical and Physical
    Sciences, Kanazawa University, Kanazawa, 920-1192, Japan,
    ogurisu@staff.kanazawa-u.ac.jp
  }%
  \and %
  {Etsuo Segawa}\thanks{ Graduate School of Information Sciences,
    Tohoku University, Aoba, Sendai 980-8579, Japan,
    e-segawa@m.tohoku.ac.jp
  }%
}

\date{August 29, 2017}

\maketitle{}

\begin{abstract}
  We show that discrete-time quantum walks on the line, $\mathbb{Z}$,
  behave as ``the quantum tunneling.''  In particular, quantum walkers
  can tunnel through a double-well with the transmission probability
  $1$ under a mild condition.  This is a property of quantum walks
  which cannot be seen on classical random walks, and is different
  from both linear spreadings and localizations.
\end{abstract}



\newpage

\section{Introduction}
\label{sec:intro}

The quantum walk (QW) is a quantum version of the classical random
walk.  Their primitive forms of the discrete-time quantum walks on
$\mathbb{Z}$ can be seen in Feynman's checker board \cite{FH}.  It is
mathematically shown (e.g.~\cite{K1}) that this quantum walk has a
completely different limiting behavior from classical random walks,
which is a typical example showing a difficulty of intuitive
description of quantum walks' behavior. 

Relations between QW and its background quantum mechanics (QM) are
very interesting, too.  
QW has been considered as quantum dynamical simulations 
such as, discretizations of the Dirac equation (e.g.~\cite{St, GF,ANF})
and also spatially discretized Schr\"{o}dinger equation (e.g.~\cite{Sh}). 
To connect QW to theses quantum dynamical system, 
some spatial and temporal scaling were needed to obtain the continuum limit
from the discrete model of QW. 
However our model treated here reproduces naturally the following 
famous quantum dynamics model following the Schr\"{o}dinger equation 
without any scaling limit. 
Here,
we consider the quantum
tunneling, which is one of the most famous quantum effects and has
well-developed since the early period of QM; this effect shows that a
quantum particle can tunnel through a barrier that it classically
could not surmount (e.g.~\cite{M1}).  In particular, the
resonant-tunneling~\cite{TE1,CET1} is very impressive; consider the
Schr{\"o}dinger equation
\begin{displaymath}
  i\frac{\partial}{\partial t}\psi(x,t)=-\frac{\partial^2}{\partial^2 x}\psi(x,t)+V(x)\psi(x,t)
  ,\quad
  \psi(\cdot,t)\in C^1(\mathbb{R})
\end{displaymath}
with a double-barrier (double-well) potential,
\begin{displaymath}
  V(x) =
  \begin{cases}
    V_0, & \mbox{ if } -L-w<x<-L \mbox{ or } L<x<L+w \\
    0, & \mbox{ otherwise. } \\
  \end{cases}
\end{displaymath}
Here, $2L>0$ is the distance between the two barriers, $w$ is the width
of them, and $V_0>0$.  Then, let us inject the the plane wave with
positive energy $E$ from $x=-\infty$.  The wave function $\psi(x,t)$
must be
\begin{displaymath}
  \psi(x,t)
  =
  e^{-iEt}
  \times
  \begin{cases}
    e^{i\sqrt{E}x} + {\rho} e^{-i\sqrt{E}x}, & \mbox{ for } x < -(L+w),\\
    {\tau} e^{i\sqrt{E}x}                 , &  \mbox{ for } x > L+w,\\
  \end{cases}
\end{displaymath}
with constants $\tau$ and $\rho\in\mathbb{C}$.  In this case, we
define the transmission probability by $T=|{\tau}|^2$ and the
reflection probability by $R=|{\rho}|^2$.  It is well known that
$T+R=1$.  In particular, it holds that $T=1$ and $R=0$ for some
resonance level, $E=E_0$.  Therefore, the injected plane wave can
tunnel the double-barrier without any reflection.  Remark that $\psi$
is not a $L^2(\mathbb{R})$-function but a bounded function.

\medskip

In this article, we show that the quantum walker can behave like the
above resonant-tunneling.  We consider a two-state QW on $\mathbb{Z}$
with two defects (double-barrier) and will prove that the quantum
walker can tunnel through the double-barrier without any reflection,
nevertheless the double-barrier has non-zero reflection elements.  We
call it {\em Quantum resonant-tunneling walk} (QRTW).

We believe that QRTW is the {\em third\,} characteristic of QW
compared to classical random walk, because any random walk on
$\mathbb{Z}$ cannot archive $R=0$ except the trivial case and it is
well-known that QW has (at least) two famous characteristics, namely,
localization and linear spreading.
The third interesting behavior cannot be derived without moving our 
concentration from the square summable space to the boundary 
functional space.
From the view point of this $\ell^\infty$ space, 
our QW model describes the solution of the quantum graph~\cite{albe,ES,GS} following the Schr{\"o}dinger equation of the
the quantum tunneling with double-barrier. See Section~\ref{sec:vs} for more detailed discussion. 
\medskip

Our organization of this paper is the following: %
In Section~\ref{sec:maintheorem}, we define our QW model with
double-barrier and prove our main result, Theorem~\ref{thm:main}. %
In Section~\ref{sec:vs}, we compare QRTW with resonant-tunneling in QM
using quantum graph walk~\cite{ES,GS}. %
In Section~\ref{sec:discussion}, we discuss our choice of initial
states, stationary measures, and experimental realization.  In
Appendix, we give a short comment on the equality, $T+R=1$.

\section{Main Theorem}
\label{sec:maintheorem}

Let us recall the definition of two-state QW model on $\mathbb{Z}$
(e.g.~\cite{K2}).  Define
\begin{displaymath}
  |L\rangle=
  \begin{bmatrix}
    1 \\ 0
  \end{bmatrix}, \quad
  |R\rangle=
  \begin{bmatrix}
    0 \\ 1
  \end{bmatrix},
\end{displaymath}
where $L$ and $R$ refer to the left and right chirality state,
respectively.  The time evolution of the walk at $x\in\mathbb{Z}$ is
determined by 2-dimensional unitary matrix
\begin{equation}
  \label{eq:localcoin}
  U_x=
  \begin{bmatrix}
    a_x & b_x \\ c_x & d_x
  \end{bmatrix}
  \in U(2).
\end{equation}
To define the dynamics of our model, we divide $U_x$ into two
matrices:
\begin{displaymath}
  P_x=
  \begin{bmatrix}
    a_x & b_x \\ 0 & 0
  \end{bmatrix},
  \quad
  Q_x=
  \begin{bmatrix}
    0 & 0 \\ c_x & d_x
  \end{bmatrix},
\end{displaymath}
with $U_x=P_x+Q_x$.  These $P_x$ and $Q_x$ represent that the walker
moves to the left and the right at $x$ at each time step,
respectively.  Let $\Psi_n$ denote the amplitude at time $n$ of the QW
on $\mathbb{Z}$:
\begin{displaymath}
  \Psi_n
  = 
  \begin{bmatrix}
    \cdots,
    \begin{bmatrix}
      \Psi_n^L(-1)\\
      \Psi_n^R(-1)\\
    \end{bmatrix},
    \begin{bmatrix}
      \Psi_n^L(0)\\
      \Psi_n^R(0)\\
    \end{bmatrix},
    \begin{bmatrix}
      \Psi_n^L(1)\\
      \Psi_n^R(1)\\
    \end{bmatrix},
    \cdots
  \end{bmatrix}',
\end{displaymath}
where $[\dots]'$ denotes the transposed operation. Then the time
evolution of the quantum walk is defined by
\begin{equation}
  \label{eq:timeevolution}
  \Psi_{n+1}(x) = P_{x+1}\Psi_n(x+1)+Q_{x-1}\Psi_n(x-1),
\end{equation}
where $\Psi_n(x)$ denotes the amplitude at time $n$ and position $x$.
Equivalently, 
\begin{displaymath}
  \begin{bmatrix}
    \Psi_{n+1}^L(x) \\
    \Psi_{n+1}^R(x)
  \end{bmatrix}
  =
  \begin{bmatrix}
    a_{x+1}\Psi_{n}^L(x+1) + b_{x+1}\Psi_{n}^R(x+1) \\
    c_{x-1}\Psi_{n}^L(x-1) + d_{x-1}\Psi_{n}^R(x-1)
  \end{bmatrix}.
\end{displaymath}
Let 
\begin{displaymath}
  \mathcal{H}
  :=\left\{
    \Psi=\{\Psi(x)\}_{x\in\mathbb{Z}}; \|\Psi\|^2:=\sum_{x\in\mathbb{Z}}(
    |\Psi^L(x)|^2
    +
    |\Psi^R(x)|^2)<\infty
    \right\}
\end{displaymath}
be the total Hilbert space of our QW.  It is well-known that
(\ref{eq:timeevolution}) defines a unitary operator $U^{(s)}$
acting on $\mathcal{H}$ satisfying that
\begin{displaymath}
  \Psi_n = (U^{(s)})^n\Psi_0
\end{displaymath}
for any $n\ge0$.

First, we consider the free case; namely,
\begin{displaymath}
  U_x =
  \begin{bmatrix}
    e^{ip} & 0 \\
    0 & e^{iq}
  \end{bmatrix}
\end{displaymath}
with constants $p$ and $q\in\mathbb{R}$ for all $x\in\mathbb{Z}$.  Let
\begin{displaymath}
  \Psi_0(0)=
  \begin{bmatrix}
    0 \\ 1
  \end{bmatrix},
  \quad
  \Psi_0(x) =
  \begin{bmatrix}
    0 \\ 0
  \end{bmatrix}
  \quad (x\ne{0})
\end{displaymath}
be an initial state $\Psi_0$.  Then a quantum walker stays at $x=0$ at
the initial time, and she, the quantum walker, moves to the position
$x=n$ at the time $n$ with
\begin{math}
  \Psi_{n}(n)=
  \begin{bmatrix}
    0 \\ e^{iqn}
  \end{bmatrix}.  
\end{math}
In contrast, if
\begin{displaymath}
  \Psi_0(0)=
  \begin{bmatrix}
    1 \\ 0
  \end{bmatrix},
  \quad
  \Psi_0(x) =
  \begin{bmatrix}
    0 \\ 0
  \end{bmatrix}
  \quad (x\ne{0}),
\end{displaymath}
she moves to the position $x=-n$ at the time $n$ with
\begin{math}
  \Psi_n(-n)=
  \begin{bmatrix}
    e^{ipn} \\ 0
  \end{bmatrix}
\end{math}.  These shows that the quantum walker freely runs over
$\mathbb{Z}$.

In this article, we mainly consider the following QW with two defects
at $x=0$ and $x=m(>0)$: let
\begin{displaymath}
  U_f=
  \begin{bmatrix}
    e^{ip} & 0 \\
    0 & e^{iq}
  \end{bmatrix},
  \quad
  U_b=
  \begin{bmatrix}
    a & b \\
    c & d 
  \end{bmatrix}
  \in U(2)
\end{displaymath}
with constants $p$, $q\in\mathbb{R}$ and $a$, $b$, $c$,
$d\in\mathbb{C}$, and
\begin{displaymath}
  U_x
  =
  \begin{cases}
    U_b, & \mbox{ if } x=0 \mbox{ or } x=m, \\
    U_f, & \mbox{ otherwise. }
  \end{cases}
\end{displaymath}
See Figure~\ref{fig:doublebarrier}.
\begin{figure}[htbp]
  \begin{center}
    \includegraphics[width=12.0cm]{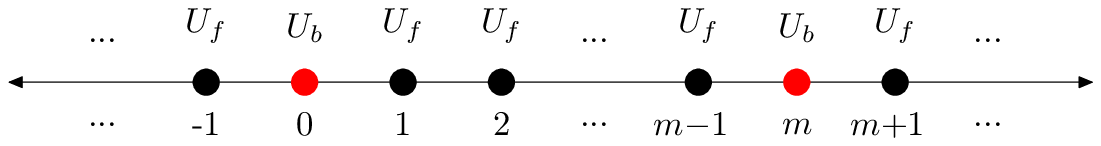}
    \caption{Two barriers $U_b$ are configured at $0$ and $m$.}
    \label{fig:doublebarrier}
  \end{center}
\end{figure}

Let
\begin{equation}
  \label{eq:initstate}
  \Psi_0(x)  =
  \begin{bmatrix}
    0 \\ e^{iqx}
  \end{bmatrix}
  \mbox{ if } x < 0,
  \quad
  \Psi_0(x) =
  \begin{bmatrix}
    0 \\ 0
  \end{bmatrix}
  \mbox{ if } x \ge 0
\end{equation}
be an initial state.  Note that $\Psi_0\not\in\mathcal{H}$ but a
bounded state and that $\Psi_n^R(x)$ is independent of $n$ for any
$x<0$.  In particular, $|\Psi_n^R(-1)|=1$, which means that one
quantum walker moves into the left barrier at $x=0$ from $x=-1$ at
each time step.  This setting corresponds to inject a plane wave into
the double-barrier from $x=-\infty$ in the resonant-tunneling
situation of QM.

Consider the infinite time limit, that is, $n\to\infty$.  Then we can
expect that $(U^{(s)})^n\Psi_0$ converges to an
$\ell^\infty$-stationary state $\Psi_\infty$ with
\begin{equation}
  \label{eq:inftimelimit}
  \Psi_\infty(x)  =
  \begin{bmatrix}
    r e^{ipx} \\ e^{iqx} 
  \end{bmatrix}
  \mbox{ if } x < 0,
  \quad
  \Psi_\infty(x) =
  \begin{bmatrix}
    0 \\ t e^{iqx}
  \end{bmatrix}
  \mbox{ if } x > m
\end{equation}
with constants $r$ and $t\in\mathbb{C}$.  We define the reflection
probability $R=|r|^2$ and the transmission probability $T=|t|^2$.
Note that $R+T=1$ (see Appendix).  Here, we define that $\phi$ is an
$\ell^\infty$-stationary state if and only if
$\phi\in\ell^\infty(\mathbb{Z};\mathbb{C}^2)$ and
$\|((U^{(s)})^n\phi)(x)\|_{\mathbb{C}^2}^2=\|\phi(x)\|_{\mathbb{C}^2}^2$
for all $x\in\mathbb{Z}$ and $n\in\mathbb{N}$.

Our main interest is the following:
\begin{quote}
  {\em Do there exist $U_f$ and $U_b$ admitting that $T=1$ ?}
\end{quote}

If there are such operators $U_f$ and $U_b$, then
(\ref{eq:inftimelimit}) means that all quantum walkers tunnel
through the double-barrier with probability $1$.  This is a
resonant-tunneling phenomenon of QW.

\begin{rem}
  The initial state $\Psi_0$ is {\em not\,} in $\mathcal{H}$, but we
  apply $U^{(s)}$ to $\Psi_0$.  Since we are interested in
  resonant-tunneling in QW, we must treat our problem in quantum
  scattering theory. (cf.\ The scattering state $\psi$ in QM in
  Section~\ref{sec:intro} is not a $L^2(\mathbb{R})$-function but a
  bounded function.) Therefore we extend the domain of $U^{(s)}$ to
  $\ell^\infty(\mathbb{Z})$ in a trivial way.  See
  Section~\ref{sec:discussion}, too.
\end{rem}

The following theorem is our main result in this article.
\begin{thm}
  \label{thm:main}
  Assume that $\Phi(0)^R=1$ and $\Phi(x)^L=0$ for $x>m$.  Let
  \begin{equation}
    \label{eq:rt}
    r=\frac{be^{ip}(1+|U_b||U_f|^{m-1})}{1-bc|U_f|^{m-1}}
    ,\quad
    t=\frac{d^2e^{-2iq}}{1-bc|U_f|^{m-1}}
  \end{equation}
  and
  \begin{equation}
    \label{eq:rtandtildes}
    \tilde{r}=\frac{re^{-ip}-b}{a},
    \quad
    \tilde{t}=\dfrac{t e^{iq(m+1)}}{d}.
  \end{equation}
  Here, $|U_f|$ and $|U_b|$ denote the determinants of $U_f$ and
  $U_b$, respectively.  Then we have that
  \begin{equation}
    \label{eq:ss}
    U^{(s)}\Phi=\Phi
    \quad\iff\quad
    \Phi(x)=
    \begin{cases}
      \begin{bmatrix}
        re^{ipx} \\ e^{iqx}
      \end{bmatrix},
      & x\le -1, \\
      \\
      \begin{bmatrix}
        \tilde{r} \\ 1
      \end{bmatrix},
      & x=0, \\
      \\
      \begin{bmatrix}
        \tilde{r}e^{ipx} \\ \tilde{t}e^{iq(x-m)}
      \end{bmatrix},
      & 0<x<m, \\
      \\
      \begin{bmatrix}
        0 \\ \tilde{t}
      \end{bmatrix},
      & x=m, \\
      \\
      \begin{bmatrix}
        0 \\ te^{iqx}
      \end{bmatrix},
      & x\ge m+1. \\
    \end{cases}
  \end{equation}
\end{thm}

\begin{proof}
  Let us prove the necessary part.  Since the quantum walker freely
  runs over $\mathbb{Z}\setminus\{0,m\}$ by $U_f$, we can write
  $\Phi(x)$ as in (\ref{eq:ss}) at
  $x\in\mathbb{Z}\setminus\{0,m\}$.  In addition, we have
  \begin{displaymath}
    \Phi(0)^{L}
    =\frac{\Phi(-1)^{L}-b\Phi(0)^{R}}{a}
    =\frac{re^{-ip}-b}{a}
    =\tilde{r}
  \end{displaymath}
  and
  \begin{displaymath}
    \Phi(m)^{R}
    =\frac{\Phi(m+1)^{R}-c\Phi(m)^{L}}{d}
    =\frac{te^{iq(m+1)}}{d}
    =\tilde{t}.
  \end{displaymath}
  Therefore, we have (\ref{eq:rtandtildes}).  Since
  $\Phi(1)^{R}=c\Phi(0)^{L}+d\Phi(0)^{R}=c\tilde{r}+d$, we have
  \begin{displaymath}
    \tilde{t}=\Phi(m)^{R}=e^{iq(m-1)}\Phi(1)^{R}=e^{iq(m-1)}[c\tilde{r}+d]. 
  \end{displaymath}
  Similarly, since
  $\Phi(m-1)^{L}=a\Phi(m)^{L}+b\Phi(m)^{R}=b\tilde{t}$, we have
  \begin{displaymath}
    \tilde{r}=\Phi(0)^{L}=e^{ip(m-1)}\Phi(m-1)^{L}=be^{ip(m-1)}\tilde{t}.
  \end{displaymath}
  Solving these simultaneous linear equations for $\tilde{t}$ and
  $\tilde{r}$, we obtain that
  \begin{displaymath}
    \begin{bmatrix}
      \tilde{t}\\
      \tilde{r}
    \end{bmatrix}
    =
    \frac{de^{iq(m-1)}}{1-bc|U_f|^{m-1}}
    \begin{bmatrix}
      1 \\ be^{ip(m-1)}
    \end{bmatrix}.
  \end{displaymath}
  This and (\ref{eq:rtandtildes}) imply (\ref{eq:rt}).
  Checking $U^{(s)}\Phi=\Phi$ by direct computations, we can easily
  prove the sufficient part.
\end{proof}

\begin{rem}
  In this theorem, since we consider the situation where quantum
  walkers are constantly injected into the double-barrier from the
  \emph{left} side, we assume that $\Phi(0)^{R}=1$ and $\Phi(x)^L=0$
  for $x>m$.  Let us consider the solution $\tilde{\Phi}$ of
  $U^{(s)}\tilde{\Phi}=\tilde{\Phi}$ with $\tilde{\Phi}(m)^L=1$ and
  $\tilde{\Phi}(x)^R=0$ for $x<0$.  This solution $\tilde{\Phi}$
  corresponds to the situation where quantum walkers are constantly
  injected into the double-barrier from the \emph{right\,} side.  Then
  we can obtain all $\ell^\infty(\mathbb{Z})$-solutions of
  $U^{(s)}\Psi=\Psi$ by linear combinations of $\Phi$ and
  $\tilde{\Phi}$.  N.~Konno, et al, have studied such
  $\ell^\infty(\mathbb{Z})$-solutions in other contexts
  in~\cite{K2,K3,K4,K5}.  See Section~\ref{sec:discussion}, too.
\end{rem}

This theorem gives us a mild condition for $T=1$, that is, QRTW.  Note
that the case where $b=c=0$ is trivial, because it is a
reflection-less case.  In the rest of this section, we omit this case.
\begin{cor}
  \label{cor:condition}
  Assume $bc\ne0$. Then, we have
  $T=1 \iff R=0 \iff |U_f|^{m-1}|U_b|=-1$.
\end{cor}

\begin{proof}
  QRTW is defined by $T=1$, equivalently, $R=0$.  Since
  \begin{displaymath}
    R=|r|^2=
    \left|\frac{be^{ip}(1+|U_b||U_f|^{m-1})}{1-bc|U_f|^{m-1}}\right|^2
  \end{displaymath}
  by Theorem~\ref{thm:main}, we obtain the desired result.
\end{proof}

\begin{rem}
  \label{rem:anotherProof}
  We can also prove this corollary using geometric series as follows.
  
  Assume $|bc|=1$.  Note that $a=d=0$ because $U_b$ is a unitary
  matrix.  Therefore, any quantum walkers is completely reflected by
  the barriers.  Thus, $T=0$.

  Assume $|bc|<1$.  Then $\Psi_\infty^R(m+1)={t}{}e^{iq(m+1)}$ is the
  summation of all the amplitudes of quantum walkers with $k$-times
  round trips between the two barriers.  Since $\Psi_n^R(-1)=e^{-iq}$
  for all $n$, we have
  \begin{displaymath}
    \Psi_\infty^R(m+1) 
    =
    \sum_{k=0}^{\infty} d{e^{iq(m-1)}} \left[ b{e^{ip(m-1)}} c{e^{iq(m-1)}} \right]^k d
    = \frac{e^{iq(m-1)}d^2}{1-bc|U_f|^{m-1}}.
  \end{displaymath}
  Write the unitary matrix $U_b$ as
  \begin{displaymath}
    U_b=
    \begin{bmatrix}
      u\overline\alpha & u\overline\beta \\
      v\beta & -v\alpha
    \end{bmatrix}
    \mbox{ with }
    \alpha, \beta, u, v\in\mathbb{C}, |\alpha|^2+|\beta|^2=|u|=|v|=1
  \end{displaymath}
  and put $e^{i\theta}=-|U_f|^{m-1}|U_b|$.  Then, we have
  \begin{equation}
    \label{eq:tau}
    \begin{split}
      |{t}|
      = \frac{1-|\beta|^2}{|1-e^{i\theta}|\beta|^2|}.
    \end{split}
  \end{equation}
  Consequently, $T=|{t}|^2=1$ if and only if
  \begin{displaymath}
    1-2|\beta|^2+|\beta|^4
    = |1-e^{i\theta}|\beta|^2|^2
    = 1-2|\beta|^2\cos\theta+|\beta|^4.
  \end{displaymath}
  Thus we have $\beta=0$ or $\cos\theta=1$.  Since the former is
  equivalent to $bc=0$, we can neglect this case by assumption.  Since
  the latter is equivalent to $-|U_f|^{m-1}|U_b|=e^{i\theta}=1$, we
  obtain the desired result.
\end{rem}

\begin{cor}
  Let $I$ be the 2-dimensional identity matrix and
  \begin{math}
    T_\theta=
    \begin{bmatrix}
      \cos2\theta & \sin2\theta \\
      \sin2\theta & -\cos2\theta
    \end{bmatrix}
  \end{math} with $\theta\in\mathbb{R}$.  Take $U_f=I$ and
  $U_b=T_\theta$. Then, $T=1$.
\end{cor}
This corollary is very important because $T_\theta$ corresponds to the
Jones matrix of a half wave plate, which is used in the implementation
of the discrete-time quantum walk by linear optical
elements~\cite{Zetal}. The detail of the implementation of QRTW will
appear in our forthcoming article. Note that $T_\theta$ includes
Hadamard matrix,
\begin{math}
  H=
  \dfrac{1}{\sqrt{2}}
  \begin{bmatrix}
    1 & 1 \\
    1 & -1 
  \end{bmatrix}
\end{math}, with $\theta=\pi/8$.

\section{QRTW vs.\ Resonant-Tunneling in QM}
\label{sec:vs}

In this section, we explain that our QRTW naturally connects the
resonant-tunneling in QM in the limit of the double-barrier width is
$0$, using the notion of the quantum graph~\cite{ES,GS}.

Let us consider the virtual quantum mechanical situation that delta
potentials~\cite{albe} are assigned on the real line $\mathbb{R}$ with
the regular interval $s>0$ at $\{sj;\,j\in \mathbb{Z}\}$ and
investigate the stationary behavior of the plane wave.  We regard it
as a ``metric'' graph whose vertices are the assigned delta potential's
places and the Euclidean length of edges are $s$.  The height of the
delta potential on $sj$ is described by $\alpha_j$ ($\geq 0$).  Let
$A$ be the set of symmetric directed edges of the one-dimensional lattice
$A=\{(j,j+1),(j+1,j);\, j\in \mathbb{Z}\}$; in $A$, we distinguish the directed edge from $j$ to $j+1$ and that from $j+1$ to $j$, 
and each directed edge has the Euclidean length $s$. 
If $a=(i,j)\in A$, then the inverse directed edge is denoted by $\bar{a}$ and the origin and terminal
vertices of $a$ are denoted by $o(a):=i$, $t(a):=j$, respectively.  
The problem can be converted to the quantum graph on this metric graph
which describes the stationary state of the plane wave on all metric
directed edges with the boundary conditions at each vertex: firstly, the domain
of the wave function is the pair of directed edge $a\in A$ and the distance
$x\in[0,s]$ from the origin vertex $o(a)$ satisfying
$\varphi(a;x)=\varphi(\bar{a};s-x)$, that is,
\begin{equation}
  \label{eq:arcs} 
  \varphi\in \{ \psi:A\times [0,s]\to\mathbb{C} \;|\; \psi(a;x)=\psi(\bar{a};s-x)\}; 
\end{equation}
secondly, the stationary Schr{\"o}dinger equation on each directed edge is 
\begin{equation}
  \label{eq:planewave} 
  -\frac{d^2}{dx^2} \varphi(a;x)=k^2\varphi(a;x); 
\end{equation}
thirdly, the boundary conditions at each vertex $u$ are given by
\begin{equation}
  \label{eq:bdc} 
  \begin{split}
    & \varphi(a;0) = \phi_{u} \mathrm{\;for \;any\;} a\in A \mathrm{\;with\;} o(a)=u; \\
    & \sum_{o(a)=u}\varphi'(a;x)|_{x=0} = \alpha_j\phi_u, 
  \end{split}
\end{equation}
where $\phi_u\in \mathbb{C}$ is an independent value of the connected
directed edge, and $\varphi'$ is the derivative of $\varphi$ with respect to
$x\in[0,s]$.  From (\ref{eq:arcs}) and (\ref{eq:planewave}),
$\varphi(a;x)$ is described by using some complex values
$\{\gamma_a\}_{a\in A}$ as follows:
\begin{equation}
  \label{eq:gamma} 
  \varphi(a;x)=\gamma_a e^{-ikx}+\gamma_{\bar{a}} e^{-ik(s-x)}. 
\end{equation}
Thus the problem is further reduced to find $\gamma_a$'s satisfying
the boundary conditions~(\ref{eq:bdc}).  The solution
$\{\gamma_a\}_{a\in A}$ satisfying all the boundary
conditions~(\ref{eq:bdc}) on all the vertices connects a quantum walk as
follows.

\begin{prop}[\cite{HKSS}] %
  \label{prop:segawa}
  Let $U_j$ be the $2$-dimensional unitary matrix of the quantum walk
  on $\mathbb{Z}$ $(j\in \mathbb{Z})$ in (\ref{eq:localcoin}) whose
  elements are given by
  \begin{displaymath}
    a_j = d_j = \frac{2e^{iks}} { 2+i\alpha_j /k },
    \quad
    b_j = c_j = e^{iks} \left(\frac{2} { 2+i\alpha_j /k } - 1\right),
  \end{displaymath}
  and $U^{(s)}$ be the total unitary operator of the quantum walk.
  Then $\{\gamma_a\}_{a\in A}$ in (\ref{eq:gamma}) is the solution
  satisfying the boundary conditions~(\ref{eq:bdc}) for all $u\in V$
  if and only if
  \[
    U^{(s)}\Psi=\Psi,
  \]
  where $\Psi(j)=[\gamma_a,\gamma_b]'$ with $t(a)=t(b)=j$ and
  $o(a)=j+1$, $o(b)=j-1$.
\end{prop}

Therefore the setting below of the following delta potential provides
the corresponding quantum tunneling walk: %
\[
  \alpha_j=
  \begin{cases}
    \alpha := wV_0 & \text{: $j=0,m$;} \\
    0 & \text{: otherwise.}
  \end{cases}
\]
Note the well known fact that the barrier potential with the width $w$
and the height $V_0$ (as in Section~\ref{sec:intro}) converges to the
delta potential $\alpha\delta(x)$ with $\alpha=wV_0$ as
$w\downarrow0$~\cite{albe}.  Thanks to Theorem~\ref{thm:main}, the
solution $\{\gamma_a\}_{a\in A}$ can be explicitly obtained.  Thus the
stationary state of our quantum tunneling walk in $\ell^\infty$ is not
only isomorphic to the stationary solution of the quantum graph
corresponding to the double-barrier delta potentials but also able to
provide the solution explicitly.  Using this, for example, we can
compute the transmission probability $T$ of the quantum graph by
\begin{equation}
  \label{eq:transProb} 
  T
  =
  |t|^2
  =
  \left(
    \frac
    {1-\frac{(\alpha/k)^2}{4+(\alpha/k)^2}}
    {\left|1+e^{2iksm}\frac{2-i\alpha/k}{2+i\alpha/k} \frac{(\alpha/k)^2}{4+(\alpha/k)^2}\right|}
  \right)^2
  =
  \left(
    \frac
    {1-|\beta|^2}
    {|1-e^{i\theta}|\beta|^2|}
  \right)^2
\end{equation}
with $|\beta|^2=\frac{4}{4+(\alpha/k)^2}$ and
$e^{i\theta}=-e^{2iksm}\frac{2-i\alpha/k}{2+i\alpha/k}$.  Therefore,
we can obtain that
\begin{equation}
  \label{eq:perfect}
  T=1 \iff e^{2iksm}\frac{2-i\alpha/k}{2+i\alpha/k}=-1
\end{equation}
in the way similar to Remark~\ref{rem:anotherProof}.  We can easily
check that (\ref{eq:perfect}) is consistent with
Corollary~\ref{cor:condition}.  Figure~\ref{fig:trans} shows the
dependence of the transmission probability $T$ on the wave number $k$
obtained by our quantum tunneling walk which is a famous figure known
as showing the quantum perfect transmission with the double-barrier,
e.g.,~\cite{albe}.
\begin{figure}[htbp]
  \begin{center}
    \includegraphics[width=7cm]{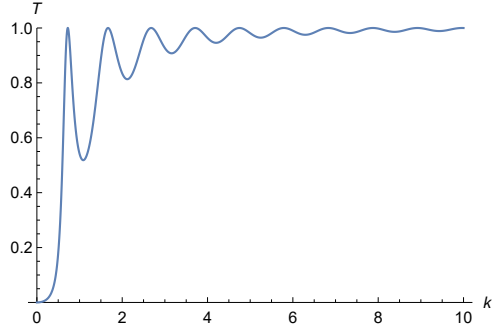}
    \caption{The wave number $k$ vs.\ the transmission probability $T$
      in (\ref{eq:transProb}) of the quantum graph for the
      $\alpha=1$ and $sm=3$ case: the wave numbers where the perfect
      transmission exhibits satisfy (\ref{eq:perfect}). }
    \label{fig:trans}
  \end{center}
\end{figure}

\section{Discussion}
\label{sec:discussion}

Our $\ell^\infty(\mathbb{Z})$-category investigations of QW in this
article have established the relation between QW and quantum scattering
theory and, in particular, revealed a resonant-tunneling phenomenon of
QW.

As already mentioned, the limit state $\Psi_\infty$ satisfies
$U^{(s)}\Psi_\infty=\Psi_\infty$ and
$\Psi_\infty\in\ell^{\infty}(\mathbb{Z})$.  Though we have considered
the initial state $\Psi_0$ defined by (\ref{eq:initstate}) in
Section~\ref{sec:maintheorem} for simplicity, it is natural to take
another initial state,
\begin{displaymath}
  \Phi_0(x)  =
  \begin{bmatrix}
    0 \\ e^{i(q+\delta)x}
  \end{bmatrix}
  \mbox{ if } x < 0,
  \quad
  \Phi_0(x) =
  \begin{bmatrix}
    0 \\ 0
  \end{bmatrix}
  \mbox{ if } x \ge 0,
\end{displaymath}
where $\delta\in\mathbb{R}$.  We can treat it in a same way as in
Section~\ref{sec:maintheorem}.  Each amplitude of $\Phi_0(x)$ gains a
phase shift $e^{i\delta}$ at each time step.  Therefore, the infinite
time limit $\Phi_\infty$ satisfies
$U^{(s)}\Phi_\infty=e^{i\delta}\Phi_\infty$.  In addition, since
\begin{displaymath}
  \Phi_\infty^R(m+1)=\frac{e^{iq(m-1)}d^2}{1-e^{i\delta(m-1)}|U_f|^{m-1}bc},
\end{displaymath}
we have that $T=1$ if and only if $(e^{i\delta}|U_f|)^{m-1}|U_b|=-1$
and $|bc|<1$.

In general, if there exists an eigenfunction $\Psi$ of $U^{(s)}$ in
$\ell^\infty(\mathbb{Z})$, then we can define a stationary measure of
the QW at position $x\in\mathbb{Z}$ by
\begin{displaymath}
  \mu(x) = |\Psi^L(x)|^2+|\Psi^R(x)|^2.
\end{displaymath}
N.~Konno, et al, have comprehensively studied such measures~\cite{K2,K3,K4,K5}.

The authors consider that such $\ell^\infty$-category studies will be
important in various areas of study of QW.

\medskip

Finally, we mention that this resonant-tunneling phenomenon of QW can
be realized in experiment. The operators will be implemented by half
wave plates and polarizing beam splitters, and the steady injection of
the quantum walker will be implemented by laser. The conceptual design
of ring-resonator named Quantum Walk Resonator will be discussed in
the forthcoming paper.

\section*{Acknowledgements}

The authors would like to thank N.~Konno for his kind discussion.
This work was supported by JSPS KAKENHI Grant Numbers %
JP17K14235, 
JP17H04978, 
JP24540208, JP16K05227, 
JP16K17637, and JP16K03939. 
KM was partially supported by Program for Promoting the reform of
national universities (Kyushu University), Ministry of Education, Culture,
Sports, Science and Technology (MEXT), Japan, World Premier
International Research Center Initiative (WPI), MEXT, Japan.

\section*{Appendix}

We use the fact that $R+T=1$ without any proof, because this is an
elementary fact derived from the two facts, the unitarity of $U^{(s)}$
and $\Psi_\infty$ being a stationary state; that is,
$U^{(s)}\Psi_\infty=\Psi_\infty$.  Consider the inflow and outflow of
$\Psi_\infty$ with respect to the interval $I: -1\le{x}\le{m+1}$.  The
quantity,
\begin{displaymath}
  P(\Psi_\infty;I) = \sum_{x=-1}^{m+1}(|\Psi_\infty^L(x)|^2+|\Psi_\infty^R(x)|^2),
\end{displaymath}
is the relative existence probability of quantum walkers in $I$.
Since $\Psi_\infty$ is stationary, $P(\Psi_\infty;I)$ is independent
of time.  On the other hand, since
$U_{-3}=U_{-1}=U_{m+1}=U_{m+3}=U_f$, we have that
$|\Psi_\infty^R(-2)|^2+|\Psi_\infty^L(m+2)|^2$ and
$|\Psi_\infty^L(-2)|^2+|\Psi_\infty^R(m+2)|^2$ are the inflow and
outflow of $I$, respectively.  Consequently, these two quantities must
be equal to each other.  The former is $1+0=1$ by the definition of
$\Psi_\infty$ and the latter is $|{r}|^2+|{t}|^2=R+T$.  Therefore,
$R+T=1$.  Note that this argument is valid even if there are more than
two barriers $U_b$ in $I$.

\end{document}